\documentclass[11pt]{article}

\usepackage[margin=1in]{geometry}
\usepackage[T1]{fontenc}
\usepackage[utf8]{inputenc}
\usepackage{lmodern}
\usepackage{amsmath,amssymb,amsfonts,amsthm,mathtools}
\usepackage{bm}
\usepackage{enumitem}
\usepackage{booktabs}
\usepackage[authoryear,round]{natbib}
\usepackage{hyperref}
\usepackage{xcolor}
\usepackage{graphicx}

\hypersetup{
    colorlinks=true,
    linkcolor=blue,
    citecolor=blue,
    urlcolor=blue
}

\newtheorem{assumption}{Assumption}
\newtheorem{definition}{Definition}
\newtheorem{proposition}{Proposition}
\newtheorem{theorem}{Theorem}

\newtheorem{corollary}{Corollary}
\newtheorem{example}{Example}
\newtheorem{remark}{Remark}

\newcommand{\R}{\mathbb{R}}
\newcommand{\Z}{\mathcal{Z}}
\newcommand{\X}{\mathcal{X}}
\newcommand{\A}{\mathcal{A}}
\newcommand{\V}{\mathcal{V}}
\newcommand{\Mcal}{\mathcal{M}_{\mathrm{cal}}}
\newcommand{\Mbf}{\mathcal{M}_{\mathrm{bf}}}
\newcommand{\M}{\mathcal{M}}
\newcommand{\Bop}{\mathcal{B}}
\newcommand{\eps}{\varepsilon}
\newcommand{\norm}[1]{\left\lVert #1 \right\rVert}
\newcommand{\abs}[1]{\left\lvert #1 \right\rvert}
\newcommand{\dd}{\mathrm{d}}

\title{Latent-Space No-Arbitrage Geometry of Generative Models for Implied Volatility Surfaces}

\author{Jing Wang, Shuaiqiang Liu\thanks{Email: shuaiqiangliu@tudelft.nl}, and Cornelis Vuik \\ \\Delft Institute of Applied Mathematics, Delft University of Technology\\}

\date{}

\begin{document}

\maketitle

\begin{abstract}
Generative models for implied volatility surfaces must produce outputs that satisfy static no-arbitrage constraints. We study these constraints in latent space. For a fixed generator, we assign each latent code a scalar margin determined by the no-arbitrage conditions of the generated surface. The codes with nonnegative margin form the admissible latent set. We establish conditions under which strictly admissible codes remain admissible under small perturbations and the boundary of the admissible set is characterized by zero margin. For regular boundary components, we formulate a level-set equation whose local dynamics are directed toward the zero-margin set. The analysis treats the generator as a map from latent variables to surfaces and is therefore not restricted to a particular architecture. It applies to variational autoencoders, generative adversarial networks, and other generative models with a deterministic realization map. Numerical tests recover known boundaries in analytic examples. Experiments with a variational autoencoder trained on Heston surfaces show that similar reconstruction errors can correspond to different admissible regions and that the latent prior may be concentrated inside such a region. The computed boundary can also be used to modify latent codes that generate violating surfaces.
\end{abstract}

\noindent\textbf{Keywords.} Implied volatility surface; no-arbitrage constraints; generative models; latent space; level-set method; Hamilton--Jacobi equation.

\section{Introduction}

Generative models are used to construct financial objects such as return scenarios, yield curves, option prices, and implied volatility surfaces \citep{Ning2023,VuleticCont2025,WangLiuVuik2025}, alongside econometric and stochastic-volatility models \citep{AitSahaliaJacod2014,Heston1993,BayerFrizGatheral2016}. For implied volatility surfaces, distributional fit alone is insufficient. A generated surface that violates static no-arbitrage conditions cannot be used consistently for pricing and hedging. Since a generator maps latent codes to entire surfaces, some parts of its latent space may produce admissible surfaces while others produce violations. We study how these two regions are separated.

Existing methods impose no-arbitrage conditions during model construction or correct violations after generation. These include admissible parameterizations \citep{ContVuletic2023,Ning2023}, constrained reconstruction maps \citep{ZhangLiZhang2023}, and penalty terms used in neural smoothing, GANs, diffusion models, and flow-matching models \citep{AckererTagasovskaVatter2020,NaZhangWan2023,VuleticCont2025,BelmonteCole2025,JinAgarwal2025,BrooksBajalicaLiuBenTahar2026}. Latent-space correction provides another possibility \citep{WangLiuVuik2025}. We keep the trained generator fixed and study the set of latent codes for which its output satisfies the prescribed no-arbitrage conditions.

\begin{figure}[!b]
    \centering
    \includegraphics[width=\textwidth]{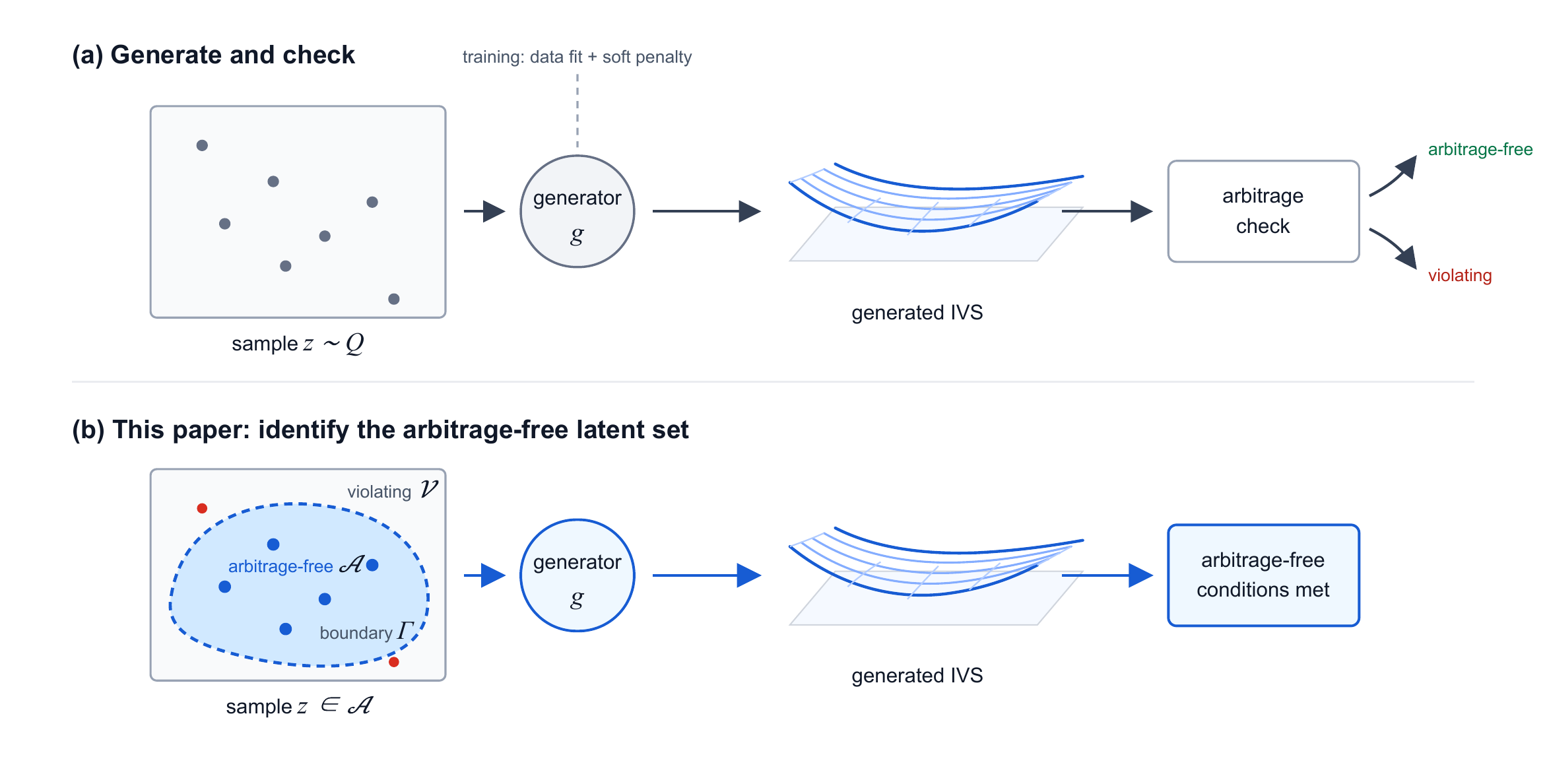}
    \caption{No-arbitrage constraints in latent space. In generate-and-check procedures, latent codes are sampled before the generated surfaces are tested for violations. The latent-space construction identifies the admissible codes before sampling from the fixed generator.}
    \label{fig:conceptual-schematic}
\end{figure}

For each latent code, we assign a scalar margin obtained from the no-arbitrage conditions of the generated surface. A positive margin corresponds to strict admissibility, and a negative margin indicates a violation. The codes with nonnegative margin form the admissible latent set. Its boundary separates latent codes that produce admissible surfaces from those that do not. This construction depends on the map from latent variables to generated surfaces, rather than on the architecture used to obtain that map. It therefore applies to VAE decoders \citep{KingmaWelling2014}, GAN generators \citep{Goodfellow2014}, normalizing flows, and other fixed generators.

We first establish that a latent code with positive margin has an admissible neighborhood. We then relate the boundary of the admissible set to the zero-margin codes and give conditions under which they agree locally. For regular boundary components, a Hamilton--Jacobi level-set equation provides an evolving-interface calculation. We test this calculation on two examples with known boundaries. We then study a two-dimensional VAE trained on Heston implied volatility surfaces by evaluating its latent margin directly and refining sign-changing grid edges. The trained-generator study compares reconstruction accuracy with the size and probability of the admissible region across several training seeds, identifies the no-arbitrage condition that determines most of the computed boundary, and tests a local correction of violating latent codes.

The remainder of this paper is organized as follows. Section~2 introduces generative models and no-arbitrage conditions of implied volatility surfaces. Sections~3--5 develop the margin, analyze the latent admissible set, and present its level-set computation. Section~6 gives toy  examples, and Section~7 presents  the results regarding trained generative  models. Section~8 concludes; the appendices contain proofs and numerical details.

\section{Generative Models and No-Arbitrage Constraints}

\subsection{Surface representation and static no-arbitrage conditions}
\label{sec:surfaces}

Let \(\Omega\subset\R^2\) denote a compact computational domain in log-moneyness and time to maturity:
\[
    (m,\tau)\in\Omega,
    \qquad m=\log(K/S_0),\quad \tau=T-t\in[\tau_{\min},\tau_{\max}],
\]
where \(0<\tau_{\min}<\tau_{\max}\), so that the short end \(\tau=0\), at which total variance degenerates, is excluded.
An implied volatility surface \citep{Dupire1994} is a map
\[
    \sigma:\Omega\to\R_+.
\]
It is often more convenient to work with total implied variance
\[
    w(m,\tau)=\tau\sigma^2(m,\tau).
\]
The use of total variance is standard in static no-arbitrage analysis because calendar-spread and butterfly-arbitrage conditions can be expressed directly in terms of \(w\) and its derivatives.

Throughout the paper, \(\X\) denotes a space of strictly positive total-variance surfaces. In the continuous formulation we use
\[
    \X=C^2_+(\Omega)
    :=\left\{w\in C^2(\Omega):\inf_{(m,\tau)\in\Omega}w(m,\tau)>0\right\},
\]
and in the discrete formulation \(\X\) is the positive cone in \(\R^{N_m\times N_\tau}\).

We equip the continuous space with the standard \(C^2\)-norm
\[
    \norm{w}_{C^2}:=\sum_{|\beta|\le2}\norm{D^\beta w}_\infty.
\]
For a total variance surface \(w\), the calendar-spread condition requires
\begin{equation}
    \partial_\tau w(m,\tau)\ge0,
    \qquad (m,\tau)\in\Omega.
    \label{eq:calendar}
\end{equation}
At a fixed maturity \(\tau\), the butterfly-arbitrage condition requires the corresponding risk-neutral density to be nonnegative. A convenient representation is given by the Gatheral--Jacquier expression \citep{GatheralJacquier2014}.
For notational simplicity, interest rates and dividends are neglected throughout, so that \(m=\log(K/S_0)\) coincides with the forward log-moneyness.
\begin{equation}
\begin{aligned}
    \Bop[w](m,\tau)
    :=&\left(1-\frac{m\partial_m w(m,\tau)}{2w(m,\tau)}\right)^2
    -\frac{(\partial_m w(m,\tau))^2}{4}
      \left(\frac{1}{w(m,\tau)}+\frac14\right) \\
    &\quad +\frac12\partial_{mm}w(m,\tau).
\end{aligned}
\label{eq:butterfly-operator}
\end{equation}
The finite-domain butterfly condition is written as
\begin{equation}
    \Bop[w](m,\tau)\ge0,
    \qquad (m,\tau)\in\Omega,
    \label{eq:butterfly}
\end{equation}
possibly together with suitable tail or boundary conditions when the full strike domain is considered. We work on finite computational domains, because generative implied volatility models typically operate on finite grids.

Accordingly, throughout this paper ``admissible'' and ``no-arbitrage'' refer to the calendar and butterfly conditions encoded on \(\Omega\) (or on the checked interior grid). They constitute a global static-arbitrage certificate only when supplemented by the appropriate strike-boundary and tail conditions. All subsequent statements are relative to the selected domain or grid and its accompanying boundary conditions.

\begin{remark}[Finite-domain viewpoint]
In empirical and generative settings, implied volatility surfaces are usually represented on a finite grid, for instance a fixed \(N_m\times N_\tau\) grid. In that case, the derivatives in \eqref{eq:calendar} and \eqref{eq:butterfly-operator} are replaced by finite-difference approximations on interior grid points. A numerical tolerance is often required to separate true arbitrage violations from discretization error.
\end{remark}

\subsection{Generative maps and induced surface families}

Let
\[
    g:\Z\subset\R^d\to\X
\]
be the total-variance realization map used in the analysis. For a latent code \(z\in\Z\),
\[
    g(z)=w_z.
\]
If an underlying decoder instead outputs implied volatility \(\sigma_z>0\), its associated realization map \(g\) is defined after the deterministic conversion
\[
    g(z)(m,\tau)=w_z(m,\tau):=\tau\sigma_z^2(m,\tau).
\]
Thus \(g\) always takes values in the total-variance space \(\X\), whether the architecture outputs total variance directly or supplies it through an implied-volatility postprocessing step. Positivity is part of the surface-space specification, not a consequence of the margin inequalities.
The definition of \(g\) is independent of the model architecture. Examples include GAN generators, normalizing-flow decoders, spline-based surface maps, VAE decoders interpreted as deterministic realization maps, and other parameterized surface-generating functions.

\begin{remark}[Deterministic generators]
All statements below concern a map \(z\mapsto w_z\) that returns one surface for each latent input. For a VAE, \(g(z)\) may be  a trained decoder; for a GAN, the random input is ordinarily included in \(z\), after which the generator is deterministic. 
\end{remark}

The generator induces a family of total variance surfaces
\[
    \{w_z:z\in\Z\}.
\]
The no-arbitrage constraints become latent-dependent inequalities:
\[
    \partial_\tau w_z(m,\tau)\ge0,
    \qquad
    \Bop[w_z](m,\tau)\ge0.
\]

Taking the minimum of the constraint values defines a scalar margin field on latent space.

\section{The No-Arbitrage Margin and the Admissible Set}
\label{sec:margin}

\subsection{No-arbitrage margins in surface and latent space}

\begin{definition}[Calendar margin]
For a sufficiently regular total variance surface \(w\), define
\begin{equation}
    \Mcal(w)
    :=\inf_{(m,\tau)\in\Omega}\partial_\tau w(m,\tau).
    \label{eq:calendar-margin}
\end{equation}
\end{definition}

\begin{definition}[Butterfly margin]
For such a surface \(w\), define
\begin{equation}
    \Mbf(w)
    :=\inf_{(m,\tau)\in\Omega}\Bop[w](m,\tau).
    \label{eq:butterfly-margin}
\end{equation}
\end{definition}

\begin{definition}[No-arbitrage margin]
The combined no-arbitrage margin is
\begin{equation}
    \M(w):=\min\{\Mcal(w),\Mbf(w)\}.
    \label{eq:combined-margin}
\end{equation}
\end{definition}

\begin{remark}[Sign convention]
We use the sign convention
\[
    \M(w)>0 \quad\text{strictly admissible for the encoded conditions},
\]
\[
    \M(w)=0 \quad\text{zero no-arbitrage margin},
\]
\[
    \M(w)<0 \quad\text{violation of an encoded condition}.
\]
\end{remark}

The margin \(\M\) is the minimum value of the no-arbitrage constraint functions over the prescribed domain. A surface is admissible only if the constraints hold at every checked grid point or throughout the continuous domain.

\begin{definition}[Latent-space no-arbitrage margin]
Given a generative map \(g:\Z\to\X\), define
\begin{equation}
    M(z):=(\M\circ g)(z)=\M(g(z)).
    \label{eq:latent-margin}
\end{equation}
\end{definition}

The composition \(M=\M\circ g\) transfers the no-arbitrage conditions from the generated surfaces to latent space. We write the latent margin as \(M\) (upright) and the surface margin as \(\M\).

\subsection{Admissible sets in continuous and discrete formulations}
\label{sec:continuous-discrete}

\begin{definition}[No-arbitrage admissible set]
The latent-space no-arbitrage admissible set is
\begin{equation}
    \A:=\{z\in\Z:M(z)\ge0\}.
    \label{eq:admissible-set}
\end{equation}
The strict admissible region is
\begin{equation}
    \A_+:=\{z\in\Z:M(z)>0\},
\end{equation}
and the violating region is
\begin{equation}
    \V:=\{z\in\Z:M(z)<0\}.
\end{equation}
The zero-margin set is
\begin{equation}
    \Sigma:=\{z\in\Z:M(z)=0\},
    \label{eq:zero-margin-set}
\end{equation}
and the topological boundary of the admissible set (relative to \(\Z\)) is
\begin{equation}
    \Gamma:=\partial_{\Z}\A.
    \label{eq:boundary}
\end{equation}
\end{definition}

If \(M\) is continuous, then \(\Gamma\subseteq\Sigma\): a boundary point cannot have strictly positive or strictly negative margin. Equality need not hold at a degenerate zero, because \(M\) may touch zero without changing sign.

A finite sample identifies admissible codes but does not determine \(\A\). The admissible set is the super-level set of \(M=\M\circ g\), with its topology and boundary defined independently of a sample.

We use continuous-domain and discrete-grid formulations.

\paragraph{Continuous-domain formulation.}
Assume \(w_z\in C^2_+(\Omega)\) and the derivatives in \eqref{eq:calendar} and \eqref{eq:butterfly-operator} are classical. Then the margins \eqref{eq:calendar-margin} and \eqref{eq:butterfly-margin} are defined by infima over \(\Omega\). Since \(\Omega\) is compact and the derivatives are continuous, these infima are attained and are in fact minima. The theoretical results use this formulation.

\paragraph{Discrete-grid formulation.}
Let
\[
    \Omega_h=\{(m_i,\tau_j):1\le i\le N_m,\;1\le j\le N_\tau\}
\]
be a finite grid. Let \(D_\tau^h\) and \(\Bop_h\) denote finite-difference approximations to \(\partial_\tau\) and \(\Bop\). Define
\begin{equation}
    \M_h(w):=\min_{(i,j)\in\mathcal I_h}
    \left\{D_\tau^h w_{ij},\,\Bop_h[w]_{ij}\right\},
    \label{eq:discrete-margin}
\end{equation}
where \(\mathcal I_h\) denotes the set of interior grid points used for derivative evaluation. Because finite-difference stencils require neighboring points, \(\mathcal I_h\) contains interior points only; arbitrage at the boundary of \(\Omega_h\) is not checked by \(M_h\) and is delegated to the tail or boundary conditions accompanying \eqref{eq:butterfly}. The discrete latent margin is
\begin{equation}
    M_h(z)=\M_h(g(z)).
\end{equation}
In numerical work, one often uses the tolerance-adjusted admissible set
\begin{equation}
    \A_h^{\eps}:=\{z\in\Z:M_h(z)\ge-\eps\},
    \label{eq:tolerance-admissible}
\end{equation}
where \(\eps>0\) accounts for finite-difference and floating-point errors.

\begin{remark}[Numerical tolerance]
The tolerance \(\eps\) should be interpreted as a pointwise numerical threshold in the arbitrage check, not as an economic relaxation of no-arbitrage. A grid point is treated as violating the condition only when the computed margin falls below \(-\eps\).
\end{remark}

With the topology fixed above, the calendar and butterfly margin functionals are continuous on \(C^2_+(\Omega)\). They are Lipschitz on sets for which \(\norm{w}_{C^2}\le R\) and \(\inf_\Omega w\ge c>0\): the differential terms are controlled by the \(C^2\)-norm, the rational terms in the butterfly operator are controlled by \(c\), and taking an infimum over compact \(\Omega\) is non-expansive in the sup-norm. When \(\X\) is a finite-dimensional positive grid space, the discrete margin \(\M_h\) is continuous on subsets bounded away from zero by the argument of Appendix~\ref{app:proofs}.

\section{Topology, Local Stability, and Boundary Regularity}
\label{sec:existence}

Under continuity of \(M\), every code with positive margin has a neighborhood contained in \(\A\). This local property does not imply admissibility throughout the latent space.

\begin{assumption}[Regularity of the generative map]
\label{ass:g-regular}
The generative map \(g:\Z\to\X\) is continuous. In stronger statements, \(g\) is assumed to be locally Lipschitz.
\end{assumption}

\begin{assumption}[Regularity of the margin functional]
\label{ass:m-regular}
The no-arbitrage margin functional \(\M:\X\to\R\) is continuous with respect to the topology on \(\X\) fixed in Section~\ref{sec:surfaces} (the \(C^2\)-norm on \(C^2_+(\Omega)\), or the Euclidean topology on the positive finite-dimensional grid space). In stronger statements, \(\M\) is assumed to be locally Lipschitz on subsets whose total variance is bounded away from zero.
\end{assumption}

\begin{assumption}[Non-degenerate boundary]
\label{ass:nondegenerate}
At an interior zero-margin point \(z_0\in\Sigma\cap\operatorname{int}(\Z)\), the latent margin is differentiable and
\[
    \nabla M(z_0)\ne0.
\]
\end{assumption}

\begin{remark}[Regularity of neural generators]
Feedforward neural networks with ReLU, leaky ReLU, tanh, sigmoid, softplus, or GELU activations are continuous and locally Lipschitz. Networks with smooth activations define smooth generators, while ReLU networks are piecewise affine and locally Lipschitz. These architectures satisfy the continuity and local Lipschitz assumptions imposed on \(g\). They do not imply that \(M=\M\circ g\) is \(C^1\), since the infimum over the surface domain and the minimum across constraints can be nonsmooth when minimizers or active constraints switch. Smoothness of the boundary requires the separate differentiability and non-degeneracy assumptions stated above.
\end{remark}

\subsection{Continuity, topology, and local stability}

\begin{proposition}[Continuity of the induced margin]
\label{prop:continuity}
Under Assumptions~\ref{ass:g-regular} and~\ref{ass:m-regular}, the latent-space margin
\[
    M=\M\circ g
\]
is continuous on \(\Z\). If both \(g\) and \(\M\) are locally Lipschitz, then \(M\) is locally Lipschitz.
\end{proposition}

\begin{proof}
The composition of continuous maps is continuous. If \(g\) is locally Lipschitz near \(z_0\) and \(\M\) is locally Lipschitz on a neighborhood of \(g(z_0)\), then for \(z,z'\) sufficiently close to \(z_0\),
\[
    |M(z)-M(z')|
    =|\M(g(z))-\M(g(z'))|
    \le L_\M \norm{g(z)-g(z')}_{\X}
    \le L_\M L_g \norm{z-z'}_{\Z}.
\]
Thus \(M\) is locally Lipschitz.
\end{proof}

\begin{theorem}[Basic topology of the admissible set]
\label{thm:topology}
Suppose \(M:\Z\to\R\) is continuous. Then
\[
    \A=\{z\in\Z:M(z)\ge0\}
\]
is closed in \(\Z\), while
\[
    \A_+=\{z\in\Z:M(z)>0\},
    \qquad
    \V=\{z\in\Z:M(z)<0\}
\]
are open in \(\Z\). If \(\Z\) is compact, then \(\A\) is compact. The set \(\A\) is nonempty if and only if there exists at least one \(z\in\Z\) with \(M(z)\ge0\); continuity alone does not guarantee nonemptiness.
\end{theorem}

\begin{proof}
The set \([0,\infty)\) is closed in \(\R\), hence \(\A=M^{-1}([0,\infty))\) is closed under continuity of \(M\). Similarly, \((0,\infty)\) and \((-\infty,0)\) are open, hence \(\A_+\) and \(\V\) are open. If \(\Z\) is compact, any closed subset of \(\Z\) is compact. The nonemptiness statement follows directly from the definition of \(\A\).
\end{proof}

\begin{theorem}[Local stability of strict admissibility]
\label{thm:local-stability}
Suppose \(M\) is continuous at \(z_0\in\Z\) and
\[
    M(z_0)>0.
\]
Then there exists \(\delta>0\) such that
\[
    z\in\Z,\quad \norm{z-z_0}<\delta
    \quad\Longrightarrow\quad
    z\in\A_+.
\]
Equivalently, relative to the latent domain,
\[
    B_\delta(z_0)\cap\Z\subset\A_+.
\]
\end{theorem}

\begin{proof}
Let \(\eta=M(z_0)/2>0\). By continuity, there exists \(\delta>0\) such that
\[
    \norm{z-z_0}<\delta
    \quad\Longrightarrow\quad
    |M(z)-M(z_0)|<\eta.
\]
Therefore
\[
    M(z)>M(z_0)-\eta=M(z_0)/2>0,
\]
so \(z\in\A_+\).
\end{proof}

\begin{corollary}[Explicit stability radius under local Lipschitz regularity]
\label{cor:lipschitz-radius}
Suppose \(M(z_0)>0\) and there exist \(r>0\) and \(L>0\) such that \(M\) is \(L\)-Lipschitz on \(B_r(z_0)\cap\Z\). Then any radius satisfying
\[
    0<\delta<\min\left\{r,\frac{M(z_0)}{L}\right\}
\]
has the property that
\[
    B_\delta(z_0)\cap\Z\subset\A_+.
\]
\end{corollary}

\begin{proof}
For \(\norm{z-z_0}<\delta\),
\[
    M(z)\ge M(z_0)-L\norm{z-z_0}>M(z_0)-L\delta>0.
\]
\end{proof}

\subsection{Boundary regularity}

\begin{theorem}[Local boundary regularity]
\label{thm:implicit}
Suppose \(M\) is \(C^1\) on a neighborhood of \(z_0\in\Sigma\cap\operatorname{int}(\Z)\) and
\[
    \nabla M(z_0)\ne0.
\]
Then \(z_0\in\Gamma\), and in a neighborhood of \(z_0\) the sets \(\Gamma\) and \(\Sigma\) coincide and form a \(C^1\) hypersurface of dimension \(d-1\), where \(d=\dim(\Z)\). Its normal direction is
\[
    n(z_0)=\frac{\nabla M(z_0)}{\norm{\nabla M(z_0)}}.
\]
\end{theorem}

\begin{proof}
The implicit function theorem applied to \(M(z)=0\) shows that \(\Sigma\) is a local \(C^1\) hypersurface. Since \(\nabla M(z_0)\ne0\), moving a sufficiently small distance from \(z_0\) in the directions \(\pm\nabla M(z_0)\) gives margins of opposite sign. Hence every nearby point of \(\Sigma\) separates \(\A\) from \(\V\), so \(\Gamma=\Sigma\) locally.
\end{proof}

\begin{remark}[Boundary points are unstable in general]
At a boundary point \(z_0\) with \(M(z_0)=0\), first-order perturbations satisfy
\[
    M(z_0+h)=\nabla M(z_0)\cdot h+o(\norm h).
\]
If \(\nabla M(z_0)\ne0\), perturbations in one normal direction enter the admissible region and perturbations in the opposite direction enter the violating region. Thus boundary codes should not be viewed as locally stable without additional margin.
\end{remark}

If the latent prior is standard Gaussian, \(z\sim N(0,I)\), the high-probability region lies near the origin. If the margin is positive throughout a central region, Theorem~\ref{thm:local-stability} gives local stability there. A path into the latent tails or beyond the observed feature domain may approach or cross \(\Gamma\). Whether violations are more frequent in these regions depends on the trained generator. Positive margin implies local admissibility, not a global no-arbitrage guarantee.

\section{Computing the Admissible-Set Boundary}
\label{sec:hj}

\subsection{Level-set representation and local attraction}

The margin provides the implicit level set
\[
    \Sigma=\{z\in\Z:M(z)=0\}.
\]
The actual boundary satisfies \(\Gamma\subseteq\Sigma\), with local equality at the non-degenerate zeros covered by Theorem~\ref{thm:implicit}. The implicit representation does not require a direct parameterization of the boundary and remains defined for multiple components, holes, corners, and changes in topology. It also yields an evolution equation that can be analyzed with viscosity-solution theory and approximated numerically.

For numerical purposes, introduce an auxiliary level-set function \citep{OsherSethian1988,Sethian1999}
\[
    \Phi:\Z\times[0,\infty)\to\R,
\]
whose zero contour
\[
    \Gamma(t)=\{z\in\Z:\Phi(z,t)=0\}
\]
represents an evolving approximation to the no-arbitrage boundary.

A natural level-set evolution is
\begin{equation}
    \Phi_t+M(z)|\nabla_z\Phi|=0.
    \label{eq:hj}
\end{equation}
The Hamiltonian is
\begin{equation}
    H(z,p)=M(z)|p|.
    \label{eq:hamiltonian}
\end{equation}
Writing \(F=-\Phi_t/\abs{\nabla\Phi}\) for the speed of the front in the direction of \(\nabla\Phi\), equation \eqref{eq:hj} gives \(F=M(z)\). With the convention \(\Phi<0\) on the admissible side and \(\Phi>0\) on the violating side, a positive margin moves the front toward the boundary from the admissible side, a negative margin moves it in the opposite direction, and the speed vanishes on \(\Sigma\). The following proposition concerns attraction along the fixed normal line through a regular zero, where \(\Sigma\) and \(\Gamma\) agree. It does not establish convergence of a curved front or coverage of every component of \(\Gamma\).

\begin{proposition}[Fixed-normal attraction at a regular zero-margin point]
\label{prop:front-convergence}
Let \(M\in C^1\) near \(z_*\in\Sigma\), with \(\nabla M(z_*)\ne0\), and set
\[
    n=-\frac{\nabla M(z_*)}{\norm{\nabla M(z_*)}}.
\]
With \(n\) pointing from the positive-margin side into the negative-margin side, for all sufficiently small \(\rho_0\) the normal trajectory defined by
\[
    \dot\rho(t)=M\bigl(z_*+\rho(t)n\bigr),
    \qquad \rho(0)=\rho_0,
\]
exists for all \(t\ge0\), remains in a neighborhood of zero, and satisfies \(\rho(t)\to0\). Moreover,
\[
    \lim_{t\to\infty}\frac{1}{t}\log|\rho(t)|
    =-\norm{\nabla M(z_*)}
\]
whenever \(\rho_0\ne0\). Thus the zero is locally attracting along this fixed normal cross-section, with asymptotic exponential rate \(\norm{\nabla M(z_*)}\). The conclusion does not assert convergence of a general characteristic or of the full level-set front.
\end{proposition}

\begin{proof}
Write \(f(\rho)=M(z_*+\rho n)\). Then \(f(0)=0\) and
\[
    f'(0)=\nabla M(z_*)\cdot n=-\norm{\nabla M(z_*)}<0.
\]
Hence \(\rho f(\rho)<0\) for all sufficiently small nonzero \(\rho\). Standard one-dimensional ODE theory gives a unique trajectory that remains in this neighborhood, is monotone toward zero, and converges to zero. Since
\[
    \frac{f(\rho)}{\rho}\longrightarrow-\norm{\nabla M(z_*)}
    \qquad\text{as }\rho\to0,
\]
integration of \(\dot\rho/\rho=f(\rho)/\rho\) gives the stated logarithmic decay rate.
\end{proof}

\begin{remark}[Local attraction versus global boundary coverage]
Proposition~\ref{prop:front-convergence} describes only a fixed-normal cross-section. A front initialized in one region need not reach every connected component of \(\Gamma\), and a bounded positive-margin component may have several boundary components. If \(M\equiv c>0\), the front propagates with constant speed and never stops; similarly, in the upper region \(z_2\ge-2/3\) of Example~2 no zero-margin interface lies ahead of the front. The initialization must therefore be chosen for the boundary components under consideration. Near a regular component, the computation monitors \(|\Phi_t|\approx|M|\,|\nabla\Phi|\) and stops when the front margin is below a tolerance of order \(\norm{\nabla M}\,h\) (Section~\ref{sec:numerical}). The proposition gives this local scale but does not prove convergence of the discretized front.
\end{remark}

\begin{remark}[Hamilton-Jacobi versus Hamilton-Jacobi-Bellman]
Equation~\eqref{eq:hj} is a Hamilton-Jacobi level-set equation with a prescribed velocity field \(M(z)\). A Hamilton-Jacobi-Bellman equation would arise if the normal velocity were selected by optimizing over controls, for example in an optimal projection or minimal-repair problem. The present level-set representation does not require such an optimization.
\end{remark}

\begin{remark}[Cost of the speed field]
A single evaluation of \(M(z)\) requires decoding an entire implied volatility surface and checking both no-arbitrage conditions on a grid; it is the dominant cost per node, not the level-set arithmetic. This cost motivates narrow-band updates, caching, and adaptive refinement, but it does not imply that an evolving front always requires fewer total margin evaluations than a one-shot full-grid contour calculation.
\end{remark}

\subsection{Well-posedness and numerical approximation}

Classical differentiability of \(\Phi\) cannot be expected in general. Even if \(M\) and \(\Phi_0\) are smooth, Hamilton-Jacobi equations may develop corners, shocks, or kinks in finite time. The appropriate solution concept is the viscosity solution \citep{CrandallLions1983}.

\begin{definition}[Viscosity solution, informal]
A continuous function \(\Phi\) is a viscosity subsolution of
\[
    \Phi_t+H(z,\nabla\Phi)=0
\]
if every smooth test function touching \(\Phi\) from above satisfies the PDE inequality at the contact point. It is a viscosity supersolution if every smooth test function touching from below satisfies the reverse inequality. A viscosity solution is both a subsolution and a supersolution.
\end{definition}

Viscosity solutions are stable under nonsmooth limits and provide the solution concept used for the level-set equation.

\begin{theorem}[Well-posedness of the level-set Cauchy problem]
\label{thm:hj-solvability}
Let the spatial domain be \(\R^d\), suppose that \(M\) is bounded and globally Lipschitz, and let \(\Phi_0\) be bounded and uniformly continuous. Then the Cauchy problem
\begin{equation}
\begin{cases}
    \Phi_t+M(z)|\nabla\Phi|=0, \\
    \Phi(z,0)=\Phi_0(z),
\end{cases}
\label{eq:hj-ivp}
\end{equation}
admits a unique bounded uniformly continuous viscosity solution. The same conclusion holds when \(M\) is Lipschitz and periodic and \(\Phi_0\) is continuous and periodic on the flat torus \(\mathbb T^d\).
\end{theorem}

\begin{proof}
The Hamiltonian \(H(z,p)=M(z)|p|\) is continuous in \((z,p)\). If \(L_M\) is the Lipschitz constant of \(M\), then
\[
    |H(z,p)-H(z',p)|\le L_M|z-z'|\,|p|,
\]
and for all \(p,q\)
\[
    |H(z,p)-H(z,q)|\le|M(z)|\,\bigl||p|-|q|\bigr|\le\norm{M}_{\infty}\,|p-q|,
\]
so \(H\) is uniformly Lipschitz in the momentum variable, uniformly in \(z\). The comparison principle of \citet{CrandallLions1983} therefore applies, and the standard existence-uniqueness theory for first-order Hamilton-Jacobi equations with continuous Hamiltonians and bounded uniformly continuous initial data yields a unique viscosity solution; see also \citet{Sethian1999} for the level-set context and Appendix~\ref{app:viscosity}.
\end{proof}

\begin{corollary}[Solvability for neural generators]
Suppose the generator \(g\) and margin functional \(\M\) induce a bounded globally Lipschitz margin \(M=\M\circ g\) on \(\R^d\) (or a Lipschitz periodic margin on \(\mathbb T^d\)). Then \eqref{eq:hj-ivp} is well posed in the viscosity sense for bounded uniformly continuous (respectively continuous periodic) initial data. On a finite latent box, one must additionally specify boundary data, impose a periodic interpretation, or fix an extension of \(M\) and \(\Phi_0\) beyond the box before invoking the Cauchy-problem result.
\end{corollary}

On a two-dimensional latent domain, let \(z_{ij}=(z_{1,i},z_{2,j})\) be a grid and let \(\Phi_{ij}^n\) approximate \(\Phi(z_{ij},t_n)\). A basic explicit update is
\begin{equation}
    \Phi_{ij}^{n+1}
    =\Phi_{ij}^{n}-\Delta t\,M_{ij}\,|\nabla_h\Phi_{ij}^n|,
    \label{eq:basic-level-set-scheme}
\end{equation}
where \(|\nabla_h\Phi|\) is computed using an upwind or Godunov numerical Hamiltonian \citep{Sethian1999}. A stability condition of the form
\begin{equation}
    \Delta t\le C\frac{\Delta z}{\max_{ij}|M_{ij}|}
    \label{eq:cfl}
\end{equation}
is typically required, where \(C\le1\) depends on the scheme and grid geometry.

Higher-order ENO or WENO discretizations can reduce numerical diffusion. Reinitialization may be used to maintain \(\Phi\) close to a signed-distance function:
\begin{equation}
    \Phi_s=\operatorname{sign}(\Phi_0)(1-|\nabla\Phi|),
    \label{eq:reinit}
\end{equation}
whose steady state satisfies \(|\nabla\Phi|=1\).

\begin{remark}[Full-grid and narrow-band costs]
On a \(d\)-dimensional latent grid with \(N\) points per axis, a full-grid contour calculation evaluates \(M\) at \(N^d\) points, each at cost \(C_M\) (one surface decode and one no-arbitrage grid check). A narrow-band level-set method stores and updates \(O(N^{d-1})\) active nodes per step, with cumulative work \(O(N^{d-1}N_t)\) for \(N_t\) time steps. Under an explicit CFL condition, a fixed physical time generally gives \(N_t=O(N)\). Hence a narrow band does not provide an unconditional asymptotic reduction relative to one \(O(N^d)\) contour calculation, and repeated updates can cost more. Direct contouring is sufficient for a static two-dimensional margin field. An evolving interface, local refinement, or restricted active memory can instead favor a level-set calculation. The total number of margin evaluations depends on travel distance, caching, initialization, and the boundary components sought.
\end{remark}

\section{Toy Experiments}
\label{sec:examples}

The toy experiments separate three roles of the latent margin: the motion of a level-set front, the induction of a financial constraint by a generator, and the interaction of several active constraints. We first use the examples with known boundaries to assess numerical reconstruction and then examine the admissible geometry of a trained generator.

\subsection{Example 1: radial admissible set}

\begin{example}[Radial margin]
Let \(\Z=\R^2\) and define
\begin{equation}
    M(z_1,z_2)=1-z_1^2-z_2^2.
    \label{eq:radial-margin}
\end{equation}
Then
\[
    \A=\{(z_1,z_2):z_1^2+z_2^2\le1\},
\]
and
\[
    \Gamma=\Sigma=\{(z_1,z_2):z_1^2+z_2^2=1\}.
\]
The level-set PDE is
\begin{equation}
    \Phi_t+(1-z_1^2-z_2^2)|\nabla\Phi|=0.
    \label{eq:radial-hj}
\end{equation}
\end{example}

If the initial interface is radial,
\[
    \Phi(z,0)=\norm z-R_0,
    \qquad 0<R_0<1,
\]
and \(\Phi(z,t)=\norm z-R(t)\), then \(|\nabla\Phi|=1\) and \(\Phi_t=-R'(t)\). Evaluating the PDE at the interface \(\norm z=R(t)\) gives
\[
    -R'(t)+1-R(t)^2=0,
\]
so
\begin{equation}
    R'(t)=1-R(t)^2.
    \label{eq:R-ode}
\end{equation}
The solution is
\begin{equation}
    R(t)=\tanh\left(t+\operatorname{arctanh}(R_0)\right),
    \label{eq:R-solution}
\end{equation}
and therefore
\[
    R(t)\to1.
\]
The evolving zero level set converges to the no-arbitrage boundary \(\Gamma=\{M=0\}\).

\begin{figure}[ht]
\centering
\includegraphics[width=\linewidth]{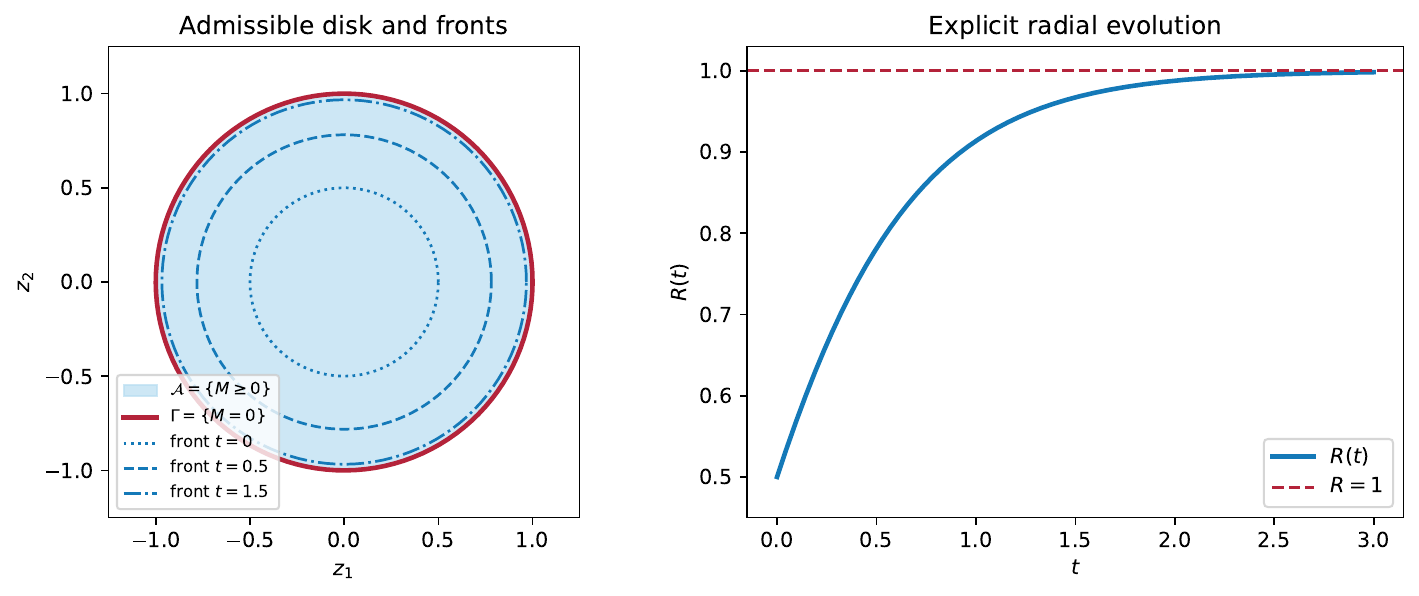}
\caption{Example 1: radial admissible set and level-set front evolution. \emph{Left}: the admissible disk \(\A=\{M\ge0\}\), its boundary \(\Gamma=\{M=0\}\), and snapshots of the front \(R(t)\to1\) starting from \(R_0=0.5\). \emph{Right}: the explicit front \(R(t)=\tanh(t+\operatorname{arctanh}R_0)\) converging to the stationary interface \(R=1\).}
\label{fig:radial}
\end{figure}

\subsection{Example 2: a generator-induced calendar margin}

This example derives the margin from a decoder-generated surface. The calendar condition is evaluated over the continuous surface domain rather than at selected points. The generated total variance remains positive on the latent box, while the calendar margin becomes negative in part of the latent tail.

Let
\begin{equation}
    \Omega=[-0.3,0.3]\times[0.1,0.6],
    \qquad
    \Z=[-6,6]^2.
\end{equation}
Consider a toy decoder that outputs total variance surfaces directly:
\begin{equation}
\begin{aligned}
    w_z(m,\tau)
    =\tau\Big[&0.08+0.004z_1
    +(0.01+0.015z_2)\tau \\
    &+0.30m^2+0.006z_1m\Big],
    \qquad z=(z_1,z_2)\in\Z.
\end{aligned}
\label{eq:numerical-toy-decoder}
\end{equation}
The corresponding implied volatility surface is
\begin{equation}
    \sigma_z(m,\tau)=\sqrt{\frac{w_z(m,\tau)}{\tau}}.
\end{equation}
The variable \(z_2\) controls the maturity dependence of the surface, while \(z_1\) affects the level and the skew-like term in \(m\). The calculation uses \(w_z\) directly because the conditions in Section~2 are expressed in total variance.

First note that the bracketed term in \eqref{eq:numerical-toy-decoder} is positive on \(\Omega\times\Z\). Indeed, the worst-case contribution from the quadratic expression in \(m\) is
\begin{equation}
    \inf_{m\in[-0.3,0.3]}\left(0.30m^2+0.006z_1m\right)
    =-\frac{(0.006z_1)^2}{4(0.30)}
    \ge -0.00108,
\end{equation}
where the unconstrained minimizer lies inside \([-0.3,0.3]\) for \(z_1\in[-6,6]\). Hence, using \(z_1=-6\), \(z_2=-6\), and \(\tau=0.6\) gives the lower bound
\begin{equation}
    0.08-0.024+(0.01-0.09)(0.6)-0.00108
    =0.00692>0.
\end{equation}
Thus the generated surfaces are well-defined as total variance surfaces throughout the latent box.

The calendar-spread derivative is
\begin{equation}
\begin{aligned}
    \partial_\tau w_z(m,\tau)
    =&\;0.08+0.004z_1
    +2(0.01+0.015z_2)\tau \\
    &+0.30m^2+0.006z_1m.
\end{aligned}
\label{eq:numerical-calendar-derivative}
\end{equation}
The decoder-induced calendar margin is therefore
\begin{equation}
    M_{\mathrm{cal}}(z)
    :=\inf_{(m,\tau)\in\Omega}\partial_\tau w_z(m,\tau).
    \label{eq:numerical-calendar-margin}
\end{equation}
The condition \(M_{\mathrm{cal}}(z)\ge0\) requires calendar admissibility for every \((m,\tau)\in\Omega\).

The infimum can be computed explicitly. The \(m\)-dependent part is convex, and its minimum is
\begin{equation}
    \inf_{m\in[-0.3,0.3]}\left(0.30m^2+0.006z_1m\right)
    =-0.00003z_1^2.
\end{equation}
The \(\tau\)-dependent part is linear. If \(0.01+0.015z_2\ge0\), equivalently \(z_2\ge-2/3\), the minimum over \(\tau\in[0.1,0.6]\) occurs at \(\tau=0.1\). Otherwise it occurs at \(\tau=0.6\). Consequently,
\begin{equation}
    M_{\mathrm{cal}}(z)=
    \begin{cases}
    0.082+0.004z_1+0.003z_2-0.00003z_1^2,
    & z_2\ge -2/3,\\[4pt]
    0.092+0.004z_1+0.018z_2-0.00003z_1^2,
    & z_2<-2/3.
    \end{cases}
    \label{eq:numerical-calendar-margin-explicit}
\end{equation}
The calendar-admissible set induced by the decoder is
\begin{equation}
    \A_{\mathrm{cal}}
    =\{z\in\Z:M_{\mathrm{cal}}(z)\ge0\}.
\end{equation}
In the upper region \(z_2\ge-2/3\), the margin is strictly positive throughout \(\Z\). For example, at the worst corner \((z_1,z_2)=(-6,-2/3)\),
\begin{equation}
    M_{\mathrm{cal}}(-6,-2/3)
    =0.082-0.024-0.002-0.00108>0.
\end{equation}
Therefore all latent codes with \(z_2\ge-2/3\) are calendar-admissible in this example.

In the lower region \(z_2<-2/3\), the boundary is obtained from \(M_{\mathrm{cal}}(z)=0\):
\begin{equation}
    z_2
    =-5.1111-0.2222z_1+0.001667z_1^2.
    \label{eq:numerical-calendar-boundary}
\end{equation}
Thus the full calendar-admissible region is
\begin{equation}
\begin{aligned}
    \A_{\mathrm{cal}}
    ={}&\{z\in[-6,6]^2:z_2\ge-2/3\} \\
    &\cup
    \left\{z\in[-6,6]^2:
    z_2<-2/3,
    \;z_2\ge -5.1111-0.2222z_1+0.001667z_1^2
    \right\}.
\end{aligned}
\label{eq:numerical-calendar-admissible-set}
\end{equation}
The violation region is the part of the latent box below this zero-level boundary:
\begin{equation}
    \V_{\mathrm{cal}}
    =\left\{z\in[-6,6]^2:
    z_2<-2/3,
    \;z_2< -5.1111-0.2222z_1+0.001667z_1^2
    \right\}.
\end{equation}
Within the box \([-6,6]^2\), the boundary enters through the left edge at approximately
\begin{equation}
    (-6,-3.7178)
\end{equation}
 and exits through the lower edge at approximately
\begin{equation}
    (4.1278,-6).
\end{equation}
Hence the central latent region remains strictly calendar-admissible, while a lower-tail region crosses the zero-margin boundary.

\begin{figure}[ht]
\centering
\includegraphics[width=0.72\linewidth]{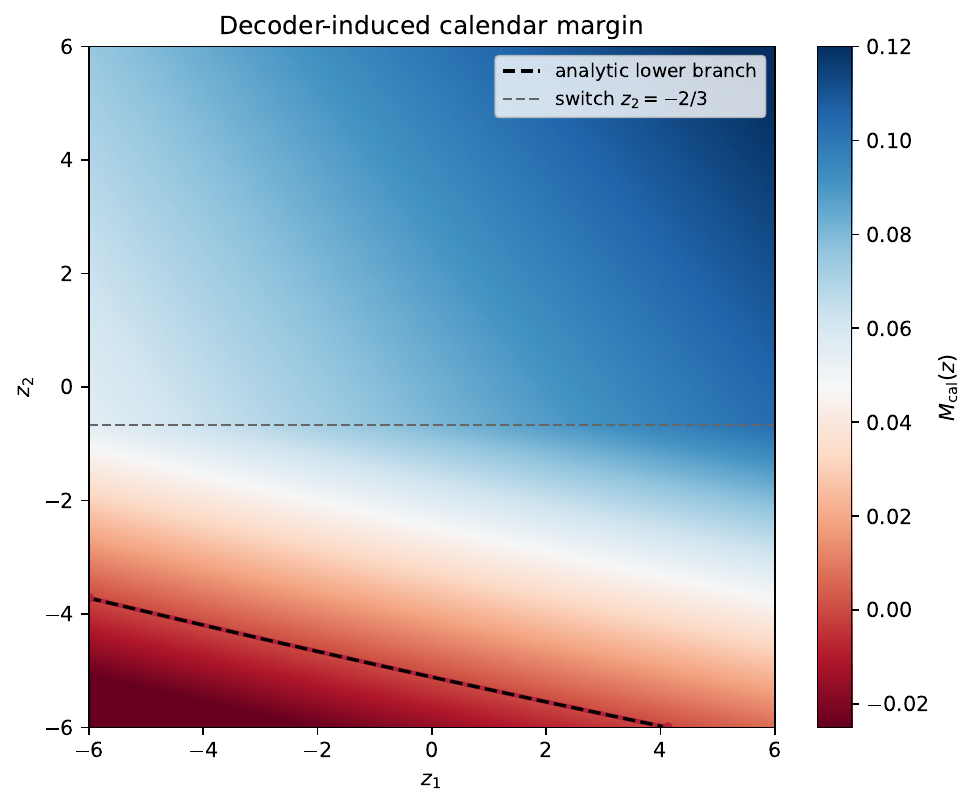}
\caption{Example 2: generator-induced calendar margin \(M_{\mathrm{cal}}(z)\) on the latent box \(\Z=[-6,6]^2\). Shading shows the margin; the red contour is the zero level set \(\partial\A_{\mathrm{cal}}\), and the dashed black curve is the analytic lower boundary branch from \eqref{eq:numerical-calendar-boundary}. The boundary enters the box at \((-6,-3.7178)\) and exits at \((4.1278,-6)\); the dashed gray line marks the switch at \(z_2=-2/3\).}
\label{fig:calendar}
\end{figure}

Table~\ref{tab:numerical-calendar-points} reports \(M_{\mathrm{cal}}\) for representative latent codes.

\begin{table}[h]
\centering
\caption{Representative latent codes for the generator-induced calendar margin in Example~2.}
\label{tab:numerical-calendar-points}
\begin{tabular}{cccc}
\toprule
\(z=(z_1,z_2)\) & \(0.01+0.015z_2\) & \(M_{\mathrm{cal}}(z)\) & Classification \\
\midrule
\((0,0)\)      & \(0.010\)  & \(0.08200\)  & strictly admissible \\
\((2,-2)\)    & \(-0.020\) & \(0.06388\)  & admissible \\
\((0,-5)\)    & \(-0.065\) & \(0.00200\)  & close to boundary \\
\((-6,-4)\)  & \(-0.050\) & \(-0.00508\) & violation \\
\((0,-6)\)    & \(-0.080\) & \(-0.01600\) & violation \\
\((6,-6)\)    & \(-0.080\) & \(0.00692\)  & admissible \\
\bottomrule
\end{tabular}
\end{table}

\paragraph{Prior mass and extrapolation probabilities.}
Let \(z\sim N(0,I_2)\) and restrict the decoder analysis to \(\Z=[-6,6]^2\). The regions \(\norm{z}\le2\) and \(\norm{z}\le3\) contain no violations for this decoder and contain respectively \(86.5\%\) and \(98.9\%\) of the Gaussian mass. A 200-node Gauss--Hermite quadrature of the analytic boundary gives the joint probability
\[
    \mathbb P\bigl(M_{\mathrm{cal}}(z)<0,\ z\in\Z\bigr)
    =3.07\times10^{-7};
\]
a fixed-seed Monte Carlo run with five million draws produced three violations, or \(6.0\times10^{-7}\) with standard error \(3.5\times10^{-7}\). The two estimates are consistent at the resolution of such a rare-event simulation.

For fixed \(z_2=c\), Table~\ref{tab:control-extrapolation} reports the analytic one-dimensional Gaussian probabilities obtained from the boundary in \eqref{eq:numerical-calendar-boundary}. They are zero for \(c\ge-3\) on the latent box, remain negligible where the boundary first enters the box, and rise sharply as \(c\) moves farther into the lower tail. These probabilities are specific to the decoder in \eqref{eq:numerical-toy-decoder}.

\begin{table}[ht]
\centering
\caption{Calendar-violation probabilities under control extrapolation in Example~2. The probability is over \(z_1\sim N(0,1)\), with the event restricted to \(z_1\in[-6,6]\) and \(z_2=c\) fixed.}
\label{tab:control-extrapolation}
\begin{tabular}{cccccc}
\toprule
\(c\) & \(-3.72\) & \(-4.00\) & \(-4.50\) & \(-5.00\) & \(-5.50\) \\
\midrule
\(\mathbb P(M_{\mathrm{cal}}(z_1,c)<0)\)
& \(5.7\times10^{-11}\) & \(7.0\times10^{-7}\) & \(3.5\times10^{-3}\) & \(3.09\times10^{-1}\) & \(9.62\times10^{-1}\) \\
\bottomrule
\end{tabular}
\end{table}

The mappings form the chain
\begin{equation}
    z\longmapsto w_z(m,\tau)
    \longmapsto \partial_\tau w_z(m,\tau)
    \longmapsto M_{\mathrm{cal}}(z)
    \longmapsto \A_{\mathrm{cal}}.
\end{equation}
The calendar condition is imposed on \(w_z\), and its minimum over \(\Omega\) defines the latent-space super-level set \(\A_{\mathrm{cal}}\).

\subsection{Numerical reconstruction of analytic boundaries}

We test the numerical level-set calculation on Examples~1 and 2, for which the relevant boundary components are available analytically. We use the first-order upwind scheme of Section~\ref{sec:numerical} on the full Cartesian grid, with the Godunov numerical Hamiltonian, homogeneous-Neumann ghost differences, and reinitialization every \(30\) steps. The iteration stops when the margin on cells crossed by the zero level set falls below a tolerance of order \(\norm{\nabla M}\,h\). Proposition~\ref{prop:front-convergence} gives this scale only for fixed-normal local dynamics and does not establish convergence of the discretized front.

For Example~1, halving \(h\) from \(4.0\times10^{-2}\) to \(2.0\times10^{-2}\) and \(1.0\times10^{-2}\) reduces the maximum radial residual from \(5.6\times10^{-2}\) to \(2.5\times10^{-2}\) and \(1.2\times10^{-2}\). On the finest grid, \(R_{\mathrm{num}}(1)=0.9129\), compared with \(R_{\mathrm{ana}}(1)=0.9137\). For Example~2, the maximum vertical residual decreases from \(1.7\times10^{-1}\) to \(1.0\times10^{-1}\) and \(4.8\times10^{-2}\) over the three grids. Both sequences are consistent with first-order refinement.

The Example~1 residual is \(|\|z\|-1|\) over all extracted zero crossings. For Example~2, the residual is \(|z_2-b(z_1)|\) for extracted points in \(-5.9<z_1<4.0\) and \(-5.9<z_2<-1.0\); it is not a symmetric Hausdorff distance. Figure~\ref{fig:levelset-ex2} compares this selected boundary component with the analytic lower branch in \eqref{eq:numerical-calendar-boundary}. All six calculations satisfy the stopping criterion. On the finest Example~2 grid, the final front band contains \(2{,}060\) of \(231{,}361\) nodes, or \(0.89\%\). This fraction records front localization; updates, margin evaluations, and memory allocation remain full-grid.

\begin{table}[ht]
\centering
\caption{Full-grid level-set reconstruction of the radial boundary in Example~1 and the selected calendar-admissibility boundary in Example~2. Residuals use the case-specific definitions and comparison window stated in the text. The front-band count reports localization within the full grid.}
\label{tab:levelset}
\begin{tabular}{lccccc}
\toprule
Case & \(N\) & \(h\) & max residual & mean residual & front band / grid \\
\midrule
Example 1 & 101 & \(4.0\times10^{-2}\) & \(5.6\times10^{-2}\) & \(5.0\times10^{-2}\) & 697 / 10{,}201 \\
Example 1 & 201 & \(2.0\times10^{-2}\) & \(2.5\times10^{-2}\) & \(2.4\times10^{-2}\) & 1{,}427 / 40{,}401 \\
Example 1 & 401 & \(1.0\times10^{-2}\) & \(1.2\times10^{-2}\) & \(1.2\times10^{-2}\) & 2{,}905 / 160{,}801 \\
Example 2 & 121 & \(1.0\times10^{-1}\) & \(1.7\times10^{-1}\) & \(9.9\times10^{-2}\) & 541 / 14{,}641 \\
Example 2 & 241 & \(5.0\times10^{-2}\) & \(1.0\times10^{-1}\) & \(6.3\times10^{-2}\) & 1{,}053 / 58{,}081 \\
Example 2 & 481 & \(2.5\times10^{-2}\) & \(4.8\times10^{-2}\) & \(3.1\times10^{-2}\) & 2{,}060 / 231{,}361 \\
\bottomrule
\end{tabular}
\end{table}

\begin{figure}[ht]
\centering
\includegraphics[width=0.62\linewidth]{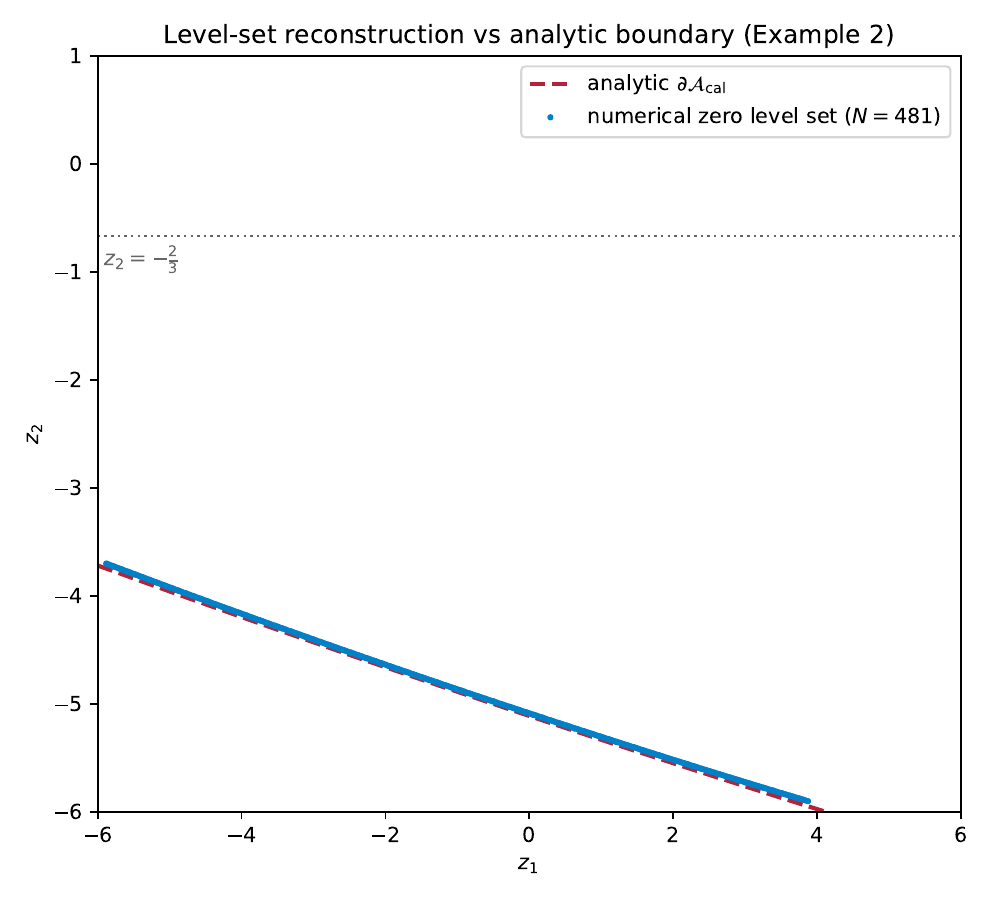}
\caption{Level-set reconstruction of the calendar boundary in Example~2 on the \(481\times481\) grid. Blue dots: numerical zero level set; dashed red curve: analytic lower branch \eqref{eq:numerical-calendar-boundary}; dotted line: switch \(z_2=-2/3\).}
\label{fig:levelset-ex2}
\end{figure}
\clearpage

\subsection{Example 3: competing calendar and butterfly constraints}

The combined admissible set is generally the intersection of multiple constraint-induced super-level sets. Consider the stylized margins
\begin{equation}
    M_{\mathrm{cal}}(z)=1-z_2^2,
    \qquad
    M_{\mathrm{bf}}(z)=1-z_1^2-\alpha z_2^2,
    \label{eq:competing-margins}
\end{equation}
where \(\alpha>0\). Let
\begin{equation}
    M(z)=\min\{M_{\mathrm{cal}}(z),M_{\mathrm{bf}}(z)\}.
\end{equation}
Then
\[
    \A=\{M_{\mathrm{cal}}\ge0\}\cap\{M_{\mathrm{bf}}\ge0\}.
\]
That is,
\begin{equation}
    \A=\{ |z_2|\le1\}\cap\{z_1^2+\alpha z_2^2\le1\}.
\end{equation}
The boundary is the union of active constraint boundaries:
\[
    \Gamma\subseteq\{M_{\mathrm{cal}}=0\}\cup\{M_{\mathrm{bf}}=0\}.
\]
The admissible set is the intersection of two convex sets and is therefore convex, but the active constraints depend on \(\alpha\). If \(\alpha\ge1\), the ellipse is contained in the strip, so the butterfly constraint alone determines the boundary (at \(\alpha=1\) the two boundaries touch at \((0,\pm1)\)). If \(0<\alpha<1\), both constraints are active: the ellipse intersects the lines \(z_2=\pm1\) at
\[
    (z_1,z_2)=\bigl(\pm\sqrt{1-\alpha},\,\pm1\bigr),
\]
and these four switching points are corners of \(\partial\A\). Thus the minimum of the two smooth margins is generally only Lipschitz across the switching set \(\{M_{\mathrm{cal}}=M_{\mathrm{bf}}\}\), while boundary nonsmoothness occurs when that set meets two active boundary pieces transversally. These stylized margins are not derived from a generator-produced surface.

\begin{figure}[ht]
\centering
\includegraphics[width=0.82\linewidth]{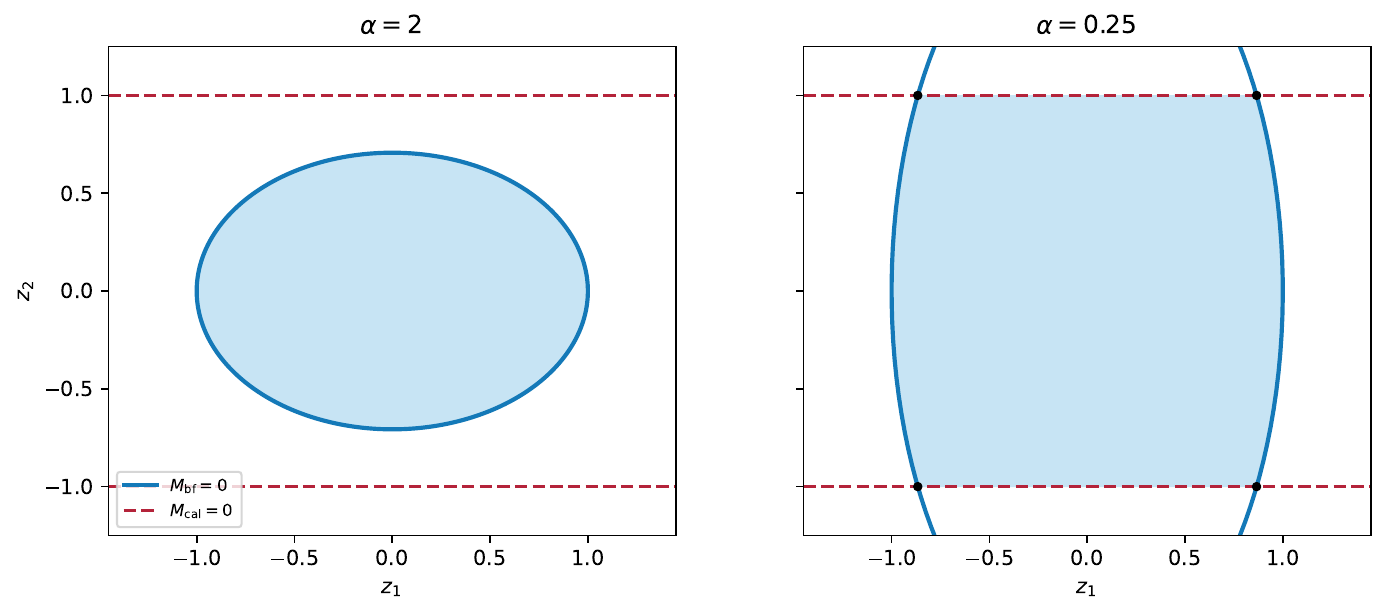}
\caption{Example 3: competing calendar and butterfly constraints. \emph{Left} (\(\alpha=2\)): the butterfly constraint dominates and the strip constraint is inactive. \emph{Right} (\(\alpha=1/4\)): both constraints shape the boundary \(\partial\A\), with switching points where \(M_{\mathrm{cal}}=M_{\mathrm{bf}}=0\).}
\label{fig:competing}
\end{figure}
\clearpage

Example~3 shows how several constraints can determine different parts of an admissible boundary. We next examine which constraints are active on the boundary induced by a trained generator, without assuming that its geometry has the stylized form above.

\section{Admissible latent geometry of Generative Financial Models} 
\label{sec:implications}
We examine the admissible latent geometry induced by a two-dimensional VAE trained on synthetic Heston implied volatility surfaces. The study compares reconstruction error, held-out admissibility, admissible area, and prior probability across training seeds.

\subsection{No-arbitrage geometry of a trained generator}

 We generate 30,000 Heston implied volatility surfaces on a grid of 28 maturities and 28 log-moneyness points. The maturities range from 0.1 to 0.6 years, and the log-moneyness values range from \(-0.27\) to \(0.27\). Each surface is converted to normalized Black--Scholes call prices with zero interest rates and dividends. We check the lower and upper call-price bounds, the two strike-slope bounds, strike convexity, and calendar monotonicity. The discrete margin is the minimum of the six constraint values over all checked grid locations. All input surfaces have positive margin, with a minimum value of \(2.79\times10^{-8}\).

We train a fully connected VAE with a two-dimensional Gaussian latent variable. The data are divided into training, validation, and test sets in proportions of 70\%, 15\%, and 15\%. The architecture and KL weight are selected using the validation set, while the test set is used only after model selection. The selected model has encoder widths 784, 512, 256, and 128, a symmetric decoder, SiLU activations, a sigmoid output, and a KL weight of 0.3 after warm-up. Five models are trained with different random seeds. For the latent-space calculations, the VAE decoder is treated as the fixed generator, and posterior means are used to represent the data in latent space.

For each seed, we evaluate the margin directly on a \(161\times161\) grid over \([-4,4]^2\). The area of the admissible region is estimated from the grid cells. Its probability under the standard Gaussian prior is estimated using 200,000 fixed-seed samples, with a Monte Carlo standard error below \(6.1\times10^{-4}\). Grid edges across which the margin changes sign are refined by bisection to a latent-coordinate tolerance of \(10^{-5}\). This static contour calculation does not use the time-dependent Hamilton--Jacobi evolution tested above.

The selected architecture reduces the mean test IV-RMSE from 0.01017 for the baseline VAE to 0.00646. Across the five selected-model seeds, the IV-RMSE varies only from 0.00639 to 0.00653. The mean fraction of strictly admissible held-out decodings is 0.974, with a range from 0.927 to 0.990. By comparison, the admissible fraction of the latent box ranges from 0.334 to 0.421, whereas the admissible probability under the Gaussian prior ranges from 0.921 to 0.981. The held-out codes and prior samples are concentrated near the central admissible region, while the uniform box measure assigns more weight to the outer region.

Seed 3 has a test IV-RMSE of 0.00645, but its held-out admissible fraction, box-area fraction, and prior admissibility are the smallest among the five seeds, at 0.927, 0.334, and 0.921. Reconstruction error therefore does not determine the no-arbitrage properties of the generator. Strike convexity is active at 2,491 of the 2,498 refined boundary points; the lower call-price bound is active at the remaining seven points. Within the checked grid on \([-4,4]^2\), each seed produces one visible admissible contour, and its boundary is determined almost entirely by strike convexity. These are finite-grid statements for the six checked call-price conditions.

\begin{table}[ht]
\centering
\caption{Reconstruction and admissibility results for the Heston VAE. The reported statistics are computed across five training seeds.}
\label{tab:heston-vae}
\begin{tabular}{lccc}
\toprule
Quantity & Mean & Standard deviation & Range \\
\midrule
Test IV-RMSE & 0.00646 & 0.00005 & 0.00639--0.00653 \\
Strictly admissible test fraction & 0.974 & 0.026 & 0.927--0.990 \\
Admissible fraction of $[-4,4]^2$ & 0.388 & 0.032 & 0.334--0.421 \\
Estimated prior admissibility & 0.966 & 0.025 & 0.921--0.981 \\
Strict boundary samples & 500 & 23 & 466--528 \\
Strike-convexity boundary share & 0.997 & 0.003 & 0.994--1.000 \\
\bottomrule
\end{tabular}
\end{table}

\begin{figure}[ht]
\centering
\includegraphics[width=\linewidth]{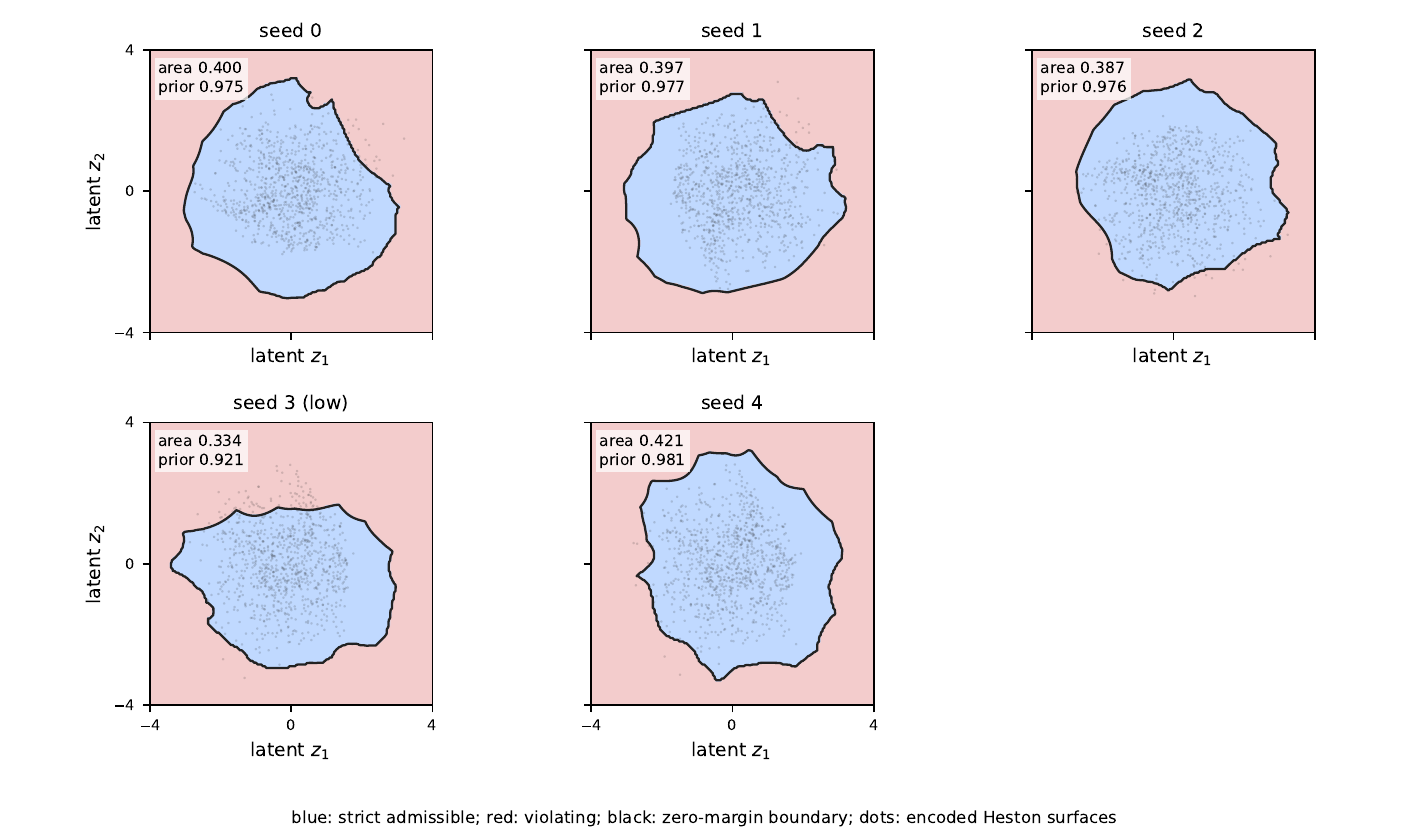}
\caption{Admissible latent regions for the Heston VAE on \([-4,4]^2\). Blue indicates strict admissibility of the generated surface, red indicates violation of at least one checked condition, and the black curve is the zero-margin boundary. Gray points are encoded Heston surfaces. The labels report the box-area fraction and the admissible probability under the standard Gaussian prior for each seed.}
\label{fig:heston-vae-geometry}
\end{figure}

\subsection{Latent-space stability}

Let \(Q\) denote a sampling distribution on latent space. The probability that the generator produces an admissible surface is \(Q(\A)\). This probability depends on both the admissible set and the mass assigned to it by \(Q\); it is different from the area of \(\A\) under a uniform measure on a bounded box. For example, the latent flow-matching study of \citet{BrooksBajalicaLiuBenTahar2026} reports that \(90.8\%\) of its generated surfaces satisfy the tested no-arbitrage conditions. That percentage estimates an admissible probability for the learned flow rather than the geometric size of its admissible set. The same distinction applies to generators based on admissible parameterizations, penalty terms, or post-generation correction \citep{ContVuletic2023,Ning2023,AckererTagasovskaVatter2020,NaZhangWan2023,BelmonteCole2025,WangLiuVuik2025}.

A global no-arbitrage guarantee requires
\[
    M(z)\ge0\qquad\text{for all } z\in\Z.
\]
Local stability concerns a neighborhood of one strictly admissible code:
\[
    M(z_0)>0
    \quad\Longrightarrow\quad
    \exists\delta>0\text{ such that }B_\delta(z_0)\cap\Z\subset\A.
\]
This implication does not determine the sign of the margin in either the central or the tail region of a latent distribution.

For a Gaussian prior, a compact truncation is required to apply the compactness statement in Theorem~\ref{thm:topology}. If a central region lies in \(\A_+\), its points are locally stable. The margin in the tails is determined by the trained generator and cannot be inferred from Gaussianity. We distinguish two sampling regions:
\begin{enumerate}[label=(\roman*)]
    \item central sampling, where margins remain positive;
    \item latent-tail sampling, where the generator is queried in low-density regions.
\end{enumerate}

For the Heston VAE, the mean admissible fraction of the fixed latent box is about 39\%, while the mean admissible probability under the Gaussian prior is about 97\%. The prior assigns most of its mass to the central admissible region. Seed 3 has one of the lowest test IV-RMSE values but the smallest admissible region in Table~\ref{tab:heston-vae}. Reconstruction error and admissible-set size therefore measure different properties of the trained generator.

Model comparisons can therefore report reconstruction error together with the admissible area on a fixed latent box and the prior probability assigned to \(\A\). Further quantities include the active constraints along the boundary and the distance of encoded data points from \(\Gamma\). These measurements describe no-arbitrage properties that are not contained in IV-RMSE.

For seed 0, we select 40 Gaussian-prior codes with \(M(z)\le-10^{-2}\) from the most negative sampled margins. Strike convexity is active at all 40 initial codes. Local ascent steps in the active finite-grid margin move 33 codes to the target margin \(2\times10^{-5}\). Across all paths, the median displacement is 0.872 and the median final margin is \(6.98\times10^{-5}\). At the endpoints, the lower call-price bound is the active constraint for 24 codes, calendar monotonicity for 9 codes, and strike convexity for 7 codes. The change in active constraints shows that a direction computed from the gradient of the currently active constraint need not remain aligned with the margin along the full path. The seven unsuccessful paths also start far from the boundary. This calculation does not solve the global projection problem in \eqref{eq:projection}.

\begin{figure}[!htbp]
\centering
\includegraphics[width=0.70\linewidth]{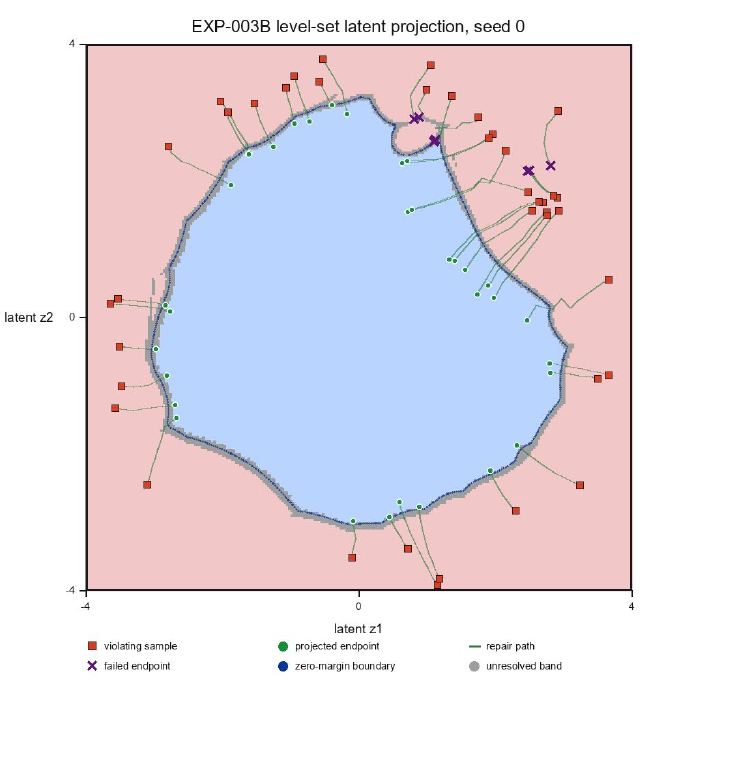}
\caption{Local correction of 40 latent codes with \(M(z)\le-10^{-2}\) for the Heston VAE, seed 0. Orange squares are the initial codes, green curves are local margin-normal paths, green points are successful endpoints, purple crosses are unsuccessful endpoints, the blue curve is the refined zero-margin boundary, and gray marks cells for which the adaptive calculation is unresolved.}
\label{fig:heston-vae-repair}
\end{figure}

If \(z\notin\A\) and \(\A\) is nonempty and closed in \(\R^d\), a natural repair problem is Euclidean projection onto the admissible set:
\begin{equation}
    \Pi_{\A}(z)
    \in\arg\min_{\tilde z\in\A}\norm{\tilde z-z}^2.
    \label{eq:projection}
\end{equation}
Equivalently,
\begin{equation}
\begin{aligned}
    \min_{\tilde z\in\Z}\quad &\norm{\tilde z-z}^2 \\
    \text{subject to}\quad & M(\tilde z)\ge0.
\end{aligned}
\label{eq:constrained-repair}
\end{equation}
In finite dimensions, nonemptiness and closedness ensure that the minimum is attained; if \(\A\) is nonconvex, the projection may nevertheless be nonunique. The level-set representation supplies the constraint geometry for this optimization problem, but it does not by itself provide a projection algorithm. In computations, \(M\) and hence the constraint are evaluated only to numerical tolerance.

\subsection{Limitations and scope}

For a trained neural generator, the margin \(M=\M\circ g\) is generally evaluated numerically. The computed boundary depends on the surface grid, derivative stencil, latent grid, and numerical tolerances. Although a non-degenerate boundary is locally smooth, the full boundary may be nonsmooth or have several connected components. The computational cost also increases with the latent dimension. Narrow-band and adaptive methods reduce active memory or improve local resolution, but they do not remove this dimensional dependence.

Local stability does not imply a global no-arbitrage guarantee. A positive margin on a discrete surface grid also does not guarantee no-arbitrage on the continuous strike--maturity domain without an approximation bound between the two formulations. The Heston results show additional variation across training seeds. Finally, if \(g(z)\) is the conditional mean of a stochastic decoder, admissibility of \(g(z)\) does not imply admissibility of every draw from \(p_\theta(\dd w\mid z)\), since the margin is nonlinear in the generated surface.

\section{Conclusion}

We studied the latent codes for which a deterministic generator produces implied volatility surfaces satisfying static no-arbitrage conditions. A scalar margin defines the admissible set and its boundary. Continuity gives local stability of strictly admissible codes, and non-degeneracy relates the boundary locally to the zero-margin set. For regular boundary components, a Hamilton--Jacobi level-set equation provides an evolving-interface calculation with local attraction toward the zero-margin set. The full-grid experiments recover the selected boundaries in the two analytic examples.

The Heston VAE produces a central admissible region for each of the five training seeds. Similar test reconstruction errors are accompanied by larger differences in admissible area and prior probability, and strike convexity determines almost all refined boundary points. The local correction experiment moves 33 of 40 violating codes across the computed boundary. These results compare reconstruction accuracy with the no-arbitrage properties of the generated surfaces.

For futher work, higher-dimensional latent spaces require more efficient boundary calculations, while a continuous-domain guarantee requires error estimates relating the discrete and continuous margins. Extensions to stochastic generators may require a probabilistic definition of admissibility.

\appendix

\section{Static No-Arbitrage Conditions for Total Variance Surfaces}

This appendix records the finite-domain conditions used in the main text. Let
\[
    w(m,\tau)=\tau\sigma^2(m,\tau).
\]
Calendar-spread no-arbitrage is expressed as
\[
    \partial_\tau w(m,\tau)\ge0.
\]
At fixed maturity \(\tau\), define
\[
    d_{\pm}(m;\tau)
    =-\frac{m}{\sqrt{w(m,\tau)}}\pm\frac{\sqrt{w(m,\tau)}}{2}.
\]
The Black-Scholes call price expressed in log-moneyness and total variance is
\[
    C_{BS}(m,w(m,\tau))
    =S_0\left(N(d_+(m;\tau))-e^mN(d_-(m;\tau))\right).
\]
A commonly used finite-domain butterfly condition is
\[
    \Bop[w](m,\tau)\ge0,
\]
where \(\Bop[w]\) is defined in \eqref{eq:butterfly-operator}. On an unbounded strike domain, additional tail conditions are needed. On a finite computational domain, one typically evaluates \(\Bop[w]\) at interior grid points and treats boundary behavior separately.

\section{Proofs and Additional Stability Results}
\label{app:proofs}

\subsection{Continuity of discrete margins}

Let \(\X_h=\R^N\) be a finite-dimensional surface representation. The finite-difference operators \(D_\tau^h\), \(D_m^h\), and \(D_{mm}^h\) are linear maps on \(\X_h\). If denominators such as \(w\) are bounded away from zero, the discrete butterfly operator \(\Bop_h\) is continuous. The minimum over finitely many grid points preserves continuity. Therefore, the discrete margin \(\M_h\) is continuous on any subset where \(w\ge c>0\). If the decoder is continuous, then \(M_h=\M_h\circ g\) is continuous.

\subsection{Closedness of tolerance-adjusted admissible sets}

For \(\eps\ge0\), define
\[
    \A_h^\eps=\{z:M_h(z)\ge-\eps\}.
\]
If \(M_h\) is continuous, then \(\A_h^\eps\) is closed. If \(M_h(z_0)>-\eps\), then \(z_0\) is locally stable with respect to the tolerance-adjusted criterion.

\section{Viscosity Solutions and Hamilton-Jacobi Equations}
\label{app:viscosity}

On \(\R^d\), or on a periodic domain, the level-set PDE
\[
    \Phi_t+H(z,\nabla\Phi)=0,
    \qquad H(z,p)=M(z)|p|,
\]
is a first-order Hamilton-Jacobi equation. Classical solutions may fail because gradients can become discontinuous. Viscosity solutions provide a weak solution concept that preserves comparison and stability. The standard theory implies uniqueness if a comparison principle holds. For bounded uniformly continuous initial data and continuous Hamiltonians satisfying suitable growth and continuity conditions, the viscosity solution exists and is stable under locally uniform perturbations of the Hamiltonian. This permits approximation of the margin field by finite differences or local surrogate evaluations.

\section{Numerical Level-Set Schemes}
\label{sec:numerical}

\subsection{Upwind discretization}

For \(\Phi_t+M(z)|\nabla\Phi|=0\), the sign of \(M\) determines the upwind direction. A monotone Godunov discretization is preferred over central differences. In two dimensions, one computes one-sided differences
\[
    D_{z_1}^-\Phi,
    \quad D_{z_1}^+\Phi,
    \quad D_{z_2}^-\Phi,
    \quad D_{z_2}^+\Phi,
\]
and combines them according to the sign of \(M_{ij}\). This prevents spurious oscillations near the interface.

\subsection{Reinitialization}

During evolution, \(\Phi\) may cease to be a signed-distance function. Reinitialization evolves
\[
    \Phi_s=\operatorname{sign}(\Phi_0)(1-|\nabla\Phi|)
\]
toward a steady state satisfying \(|\nabla\Phi|=1\). Care is needed because reinitialization can slightly move the zero level set if implemented aggressively.

\subsection{Narrow-band strategy}

A narrow-band method restricts updates and margin evaluations to points satisfying
\[
    |\Phi(z,t)|\le r_{\mathrm{band}}.
\]
When evaluation of \(M(z)=\M(g(z))\) requires decoding a full surface and computing its no-arbitrage margins, the narrow band restricts these evaluations to a neighborhood of the estimated boundary \(\Gamma\).

\subsection{Adaptive refinement}

Adaptive refinement can be guided by sign changes of \(M\), large gradients of \(\Phi\), or near-zero margin values. In two-dimensional latent spaces, marching-squares contour extraction can recover the approximate zero contour. In higher dimensions, one may use local charts, sparse grids, active subspaces, or reduced-dimensional slices.

\end{document}